\documentclass[aip,
 amsmath,amssymb,
preprint,%
]{revtex4-1}

\usepackage{graphicx}
\usepackage{dcolumn}
\usepackage{bm}

\usepackage[utf8]{inputenc}
\usepackage[T1]{fontenc}
\usepackage{mathptmx}
\usepackage{etoolbox}

\usepackage{amssymb}
\usepackage{amsmath}
\usepackage{amsthm}
\usepackage{bm}

\newtheorem{theorem}{Theorem}
\makeatletter
\def\@email#1#2{%
 \endgroup
 \patchcmd{\titleblock@produce}
  {\frontmatter@RRAPformat}
  {\frontmatter@RRAPformat{\produce@RRAP{*#1\href{mailto:#2}{#2}}}\frontmatter@RRAPformat}
  {}{}
}%
\makeatother
\begin{document}

\preprint{AIP/123-QED}

\title[Non-normality, optimality and synchronization]{Non-normality, optimality and synchronization}
\author{Jeremie Fish}
 \altaffiliation[Also at ]{Clarkson Center for Complex Systems Science}
\author{Erik Bollt}
\altaffiliation[Also at ]{Clarkson Center for Complex Systems Science}
 \email{jafish@clarkson.edu}
\affiliation{ 
Department of Electrical and Computer Engineering, Clarkson University
}%

\date{\today}

\begin{abstract}
It has been recognized for quite some time that for some matrices the spectra are not enough to tell the complete story of the dynamics of the system, even for linear ODEs. While it is true that the eigenvalues control the asymptotic behavior of the system, if the matrix representing the system is non-normal, short term transients may appear in the linear system. Recently it has been recognized that since the matrices representing directed networks are non-normal, analysis based on spectra alone may be misleading. Both a normal and a non-normal system may be stable according to the master stability paradigm, but the non-normal system may have an arbitrarily small attraction basin to the synchronous state whereas an equivalent normal system may have a significantly larger sync basin. This points to the need to study synchronization in non-normal networks more closely. In this work, various tools will be utilized to examine synchronization in directed networks, including pseudospectra, an adaption of pseudospectra that we will call Laplacian pseudospectra.  We define a resulting concept that we call Laplacian pseudospectral resilience (LPR). It will be shown that LPR outperforms other scalar measures for estimating the stability of the synchronous state to finite perturbations in a class of networks known as optimal networks. Finally we find that the ideal choice of optimal network, with an eye toward synchronization, is the one which minimizes LPR. 
\end{abstract}

\maketitle

\begin{quotation}
Synchronization, thanks to its prevalence in real world systems, has attracted a significant amount of interest. A number of tools have been developed for analysis of this phenomena, with perhaps the most widely used being the master stability function (MSF). Inspired by its discovery, the idea of optimizing network structure to maximize the synchronizability with respect to the coupling strength was born, and it has been shown that an entire class of directed networks exist which attain the optimal value \cite{nishikawa2006a}. However, MSF analysis is inherently local, which limits the ability to draw conclusions about a systems stability against large perturbations and even sometimes small ones as well \cite{menck2014} for many systems. Interestingly, directed networks generally have non-normal graph Laplacians, which frequently leads to a shrinking of the synchronization basin, further limiting the usefulness of the MSF in this context. Below we discuss non-normality, how it relates to synchronization, and how we can use tools such as pseudospectra in non-normal systems to help us understand the stability of optimal networks against finite sized perturbations. 
\end{quotation}

\section{Introduction}
Many systems evolve guided by the interaction of a large number of constituent parts \cite{newman2018}. Often a property of such systems that is of great interest, is whether or not the system synchronizes \cite{arenas2008,pikovsky2001}. Synchronization is a property which shows up in many contexts, including ecological \cite{blasius2000}, chemical \cite{epstein1996}, neurological \cite{javedan2002}, and mechanical \cite{nair2008}, among many others.

The stability of the synchronous state is frequently of concern. A stable synchronous state is sometimes highly desirable, such as in the electrical power grid \cite{blaabjerg2006,menck2014,nishikawa2015} or undesirable such as in the portion of the brain affected in Parkinson's disease \cite{chen2007, hammond2007}. For the simplest scenario, one in which all oscillators are identical, Pecora and Carroll showed that the linear stability of the synchronous state could be determined \cite{pecora1998}. This analysis was expanded by Nishikawa and Motter \cite{nishikawa2006a,nishikawa2006b} for the case of directed networks, and a new class of networks which they called {\it optimal} networks were studied. 

However this analysis  does not explain some of the behavior observed in directed networks. For instance it has been observed that larger synchronization errors can result than would be expected in such systems \cite{illing2002,jadbabaie2013}.
Recently it has also been suggested that for directed networks, which are often represented by non-normal matrices, linear stability analysis may not be enough, as non-normal matrices can exhibit transients, which may even destroy a linearly stable state \cite{fish2017,asllani2018,zankoc2019,muolo2021} when the perturbations are finite. This point has recently generated some interest and disagreement \cite{nishikawa2021comment, muolo2021comment, sorrentino2022comment}. For states arbitrarily close to the synchronization manifold, so long as the manifold has a finite basin of attraction, the lower bound on the size of the transient can be arbitrarily large for non-normal matrices \cite{trefethen2005}. However any finite system remains stable for infinitesimal perturbations away from the synchronization manifold. Of note in real world systems as verified by experiment, noise plays an important role and so finite perturbations cannot be ignored \cite{blakely2000}. This in a sense makes the linear stability analysis misleading, there can potentially be a large range of coupling strengths for which a non-normal system is linearly stable, and yet for almost every initial condition is repelled away from the synchronous state. Though clearly noise plays an important role, it will not be our focus. 

In this work we will examine the role that non-normality plays in the destabilization of the linearly stable state. We will discuss several scalar quantities related to non-normal matrices and show how they fail to properly characterize which non-normal matrices are most stable to small perturbations. This analysis leads naturally to $\epsilon$-pseudospectra, which we will adapt to the synchronization question by defining what we call $\epsilon$-Laplacian pseudospectra. We also develop a new scalar ``score" for stability, which we call Laplacian pseudospectral resilience. Through numerical evidence, we argue that this scalar outperforms other measures by correctly identifying which non-normal and optimal Laplacian matrices are most stable.
\section{An Example}
To motivate the discussion that follows, we begin with a simple example. For asymptotic stability analysis of a fixed point of a nonlinear system, generally the system is linearized about the fixed point, and declared asymptotically stable if the eigenvalues of the Jacobian matrix are all negative. Such analysis is possible thanks to the Hartman-Grobman theorem \cite{perko2001}. The Jacobian represents a linearization of the dynamics, and is only valid "nearby" to the fixed point. By comparison with a truly linear system, the fixed point should be stable if the eigenvalues are negative. However as we shall see, there is a problem with this logic, it ignores the fact that the fixed point of the system may have a finite basin of attraction. If the Jacobian (i.e. the linearization) is a {\it normal} matrix then the system will monotonically approach the fixed point, so long as it is initially in the attractive region. However for non-normal systems, the attractive region may decrease in size relative to a normal system under the same dynamics, and example of this is shown in Fig. \ref{fig:LinFail}. The difference between normal and non-normal systems an arbitrarily large transient, on an arbitrarily small timescale \cite{trefethen2005} may occur in the latter, meaning an initial condition which starts off in the region of attraction may exit due to the transient.
\begin{figure}[htbp]
\includegraphics[width=0.95\textwidth]{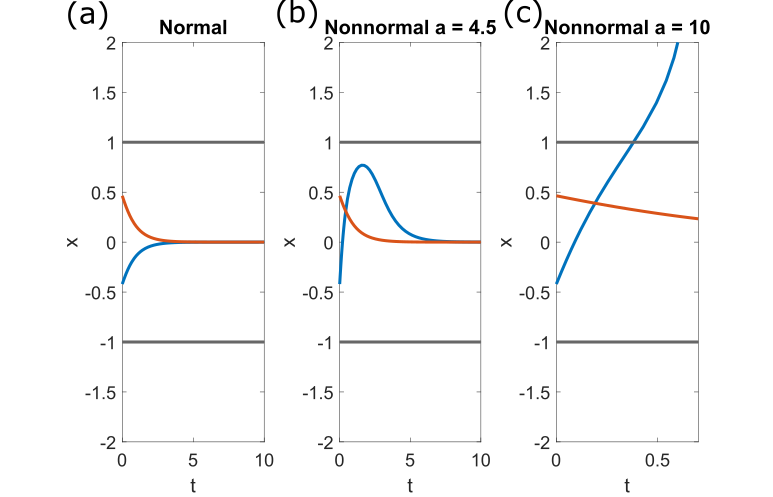}
\caption{How Jacobian linearization can be decieving. In (a),(b),and (c) the same initial condition is used, though in (a) the Jacobian about the fixed point $0$ is a normal matrix, in (b) and (c) the Jacobian is non-normal. \label{fig:LinFail}}
\end{figure}
In Fig. \ref{fig:LinFail} (a)-(c) the system is governed by the following differential equations:
\begin{equation}
    \begin{bmatrix}
    \dot{x_1} \\
    \dot{x_2}
    \end{bmatrix} = 
    \begin{bmatrix}
    x_1^3 -x_1 \\
    x_2^3 - x_2
    \end{bmatrix} + A \begin{bmatrix}
    x_1 \\
    x_2
    \end{bmatrix}. \label{eq:simplecubic}
\end{equation}
For Fig. \ref{fig:LinFail} (a) the matrix
\begin{equation}
    A = \begin{bmatrix}
    -0.1 & 0.05 \\
    0.05 & -0.1
    \end{bmatrix}, \label{eq:Normal1}
\end{equation}
is chosen, 
and for (b) and (c), 
\begin{equation}
    A = \begin{bmatrix}
    -0.1 & a \\
    0 & -0.1
    \end{bmatrix}, \label{eq:Nonnormal1}
\end{equation}
is chosen,
where $a = 4.5$ in (b) and $a = 10$ in (c).

The Jacobian of the system in Eq. \ref{eq:simplecubic} with coupling matrix Eq. \ref{eq:Normal1} about the fixed point $0$ is,
\begin{equation}
    \mathcal{J} = \begin{bmatrix}
    -1.1 & 0.05 \\
    0.05 & -1.1
    \end{bmatrix},
\end{equation}
and has all negative eigenvalues. Similarly for (b) and (c) the Jacobian is,
\begin{equation}
\mathcal{J} = 
\begin{bmatrix}
-0.1 & a \\
    0 & -0.1
\end{bmatrix}.
\end{equation}

In the uncoupled system (i.e. $a = 0$), initial conditions in the open set $(-1,1)^2$ are attracted to the fixed point at $0$. When the coupling matrix is normal, so is the Jacobian of the system and as can be seen in Fig. \ref{fig:LinFail} (a), the initial condition approaches $0$ exponentially fast. However in the scenario where the coupling matrix is non-normal, transient growth is observed. If the transient is not too large the system remains attracted to the fixed point at $0$ Fig \ref{fig:LinFail} (b). In Fig. \ref{fig:LinFail} (c) the transient is large enough that the initial condition is knocked out of the attractive region. 

\section{Non-normality}
Most of the analysis which follows below has appeared elsewhere, such as \cite{trefethen2005} for non-normality and pseudospectra in general, and specifically for networks in \cite{asllani2018,asllani2018top,muolo2019,muolo2021}.
We assume here that all matrices in this work are square unless otherwise noted. 
A matrix $A \in \mathbb{R}^{n \times n}$ is defined to be normal if \cite{horn1985,golub2013},
\begin{equation}
    A^TA = AA^T.
\end{equation}
In this work we will be examining real valued matrices, however the property of normality can be extended to complex valued matrices by substituting the conjugate transpose for the transpose and Hermitian is substituted for symmetric. It is easy to see that a matrix will be normal if it is symmetric, since $A = A^T$ in this scenario, other classes of matrices are normal as well, such as skew-symmetric, orthogonal and unitary. A matrix is called non-normal if it is not normal. Two matrices, $A$ and $B$ are called similar if there exists an invertible matrix P such that \cite{horn1985,golub2013}:
\begin{equation}
    B = P^{-1}AP.
\end{equation}
By the spectral theorem \cite{halmos1963}, any normal matrix is similar to a diagonal matrix $\Lambda$ by an orthogonal matrix $Q$, so,
\begin{equation}
    \Lambda(A) = Q^T A Q. \label{eq:Diag}
\end{equation}
Here $\Lambda(A)$ contains the eigenvalues of $A$ and this means that for normal matrices, the eigenvectors will be orthogonal to one another. For non-normal matrices we may lose orthogonality of the eigenvectors. In fact in general non-normal matrices are not even guaranteed to be similar to a diagonal matrix (though as we will see, some non-normal matrices ARE diagonalizable), so we must replace $\Lambda(A)$, with $J(A)$, which is a matrix in Jordan normal form or Jordan canonical form. So Eq. \ref{eq:Diag} becomes replaced by \cite{horn1985,golub2013},
\begin{equation}
    J(A) = S^{-1}AS.
\end{equation}
Thus non-normal matrices do not in general have orthogonal eigenvectors. This plays a role in the corresponding linear dynamics,
\begin{equation}
    \dot{x} = Ax. \label{eq:SimpLin}
\end{equation}
Notice that we may simply choose to change coordinates to examine Eq. \ref{eq:SimpLin}, in that scenario the system becomes:
\begin{equation}
    \dot{x} = SJ(A)S^{-1} x,
\end{equation}
and choosing $\eta = S^{-1} x$ we have,
\begin{equation}
    \dot{\eta} = J(A) \eta. \label{eq:transformed}
\end{equation}
If $A$ is normal, then $J(A)$ will simply be a diagonal matrix, and the coordinates will simply be changed to another (or perhaps even the same) orthogonal basis. If $A$ is non-normal, the new basis may {\it not} be orthogonal, and $J(A)$ may not be diagonal. The transient behavior of a non-normal linear system is related to both of these facts as will become clear. 

One aspect which plays a role in the size of the transient is the size of the largest Jordan block. As can be seen in Appendix \ref{ap:1} (and was previously shown in \cite{nishikawa2006a}), the solution to each Jordan block Eq. \ref{eq:transformed} will be proportional to $t^{k-1}e^{\lambda t}$, where $k$ is the size of the Jordan block and $\lambda$ is the eigenvalue of the associated block. Clearly then, the larger the Jordan block, the larger the transient will be, as the polynomial term dominates the early behavior of the system. This concept underlies the  concept of optimal networks of Nishikawa and Motter  \cite{nishikawa2010} and others \cite{ravoori2011} using the size of the largest Jordan block \cite{fish2017}, where the {\it sensistivity index} ($\mathbb{S} \in \mathbb{N}$) is defined as the size of the largest Jordan block. This idea was successfully used to determine which {\it unweighted} optimal networks would be most likely to desynchronize, something that the master stability function \cite{pecora1998} is incapable of determining since it only examines conditions infinitesmally close to the synchronization manifold.   

However it is not enough to simply find $\mathbb{S}$, as non-normal matrices which have the same eigenvalues and same Jordan form may have different sized transients as seen in Fig. \ref{fig:TwobyTwo}. In this case we have:
\begin{equation}
    B_\alpha = 
    \begin{pmatrix}
    -1 & \alpha\\
    0 & -1
    \end{pmatrix}. \label{eq:Sample2by2}
\end{equation}
The Jordan form of this matrix stays the same for all $\alpha \neq 0$, and yet the transients get larger as $\alpha \rightarrow \infty$. Similarly, despite being non-normal, and having  $\mathbb{S} = 2$, as $\alpha < \alpha^*$ we can see that no transient at all may be observed. So by changing $\alpha$, we are not changing the Jordan form (unless of course we set $\alpha =0$), but rather we are changing the basis $S$.
\begin{figure}[htbp]
\includegraphics[width=0.5\textwidth]{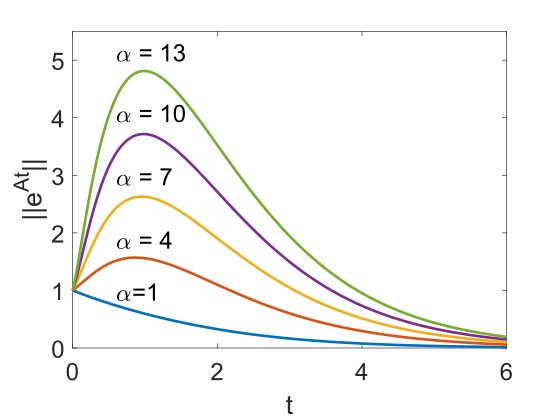}
\caption{Linear dynamics of the matrix in Eq. \ref{eq:Sample2by2} for $\alpha \in \{1,4,7,10,13 \}$. As $\alpha \rightarrow \infty$ the size of the transient grows, despite having the same Jordan form. Note that for small enough $\alpha$, despite the matrix being non-normal, no transient growth is observed. \label{fig:TwobyTwo}}
\end{figure}
So in this example we must pair information about the largest Jordan block with some additional information to understand the size of the transient. This example seems to suggest that we should look to the term $\alpha$, and Henrici's concept of departure from normality \cite{henrici1962} (or Henrici's index) which does just that. To calculate Henrici's index ($H_{\mbox{ind}}$) one must first perform the Schur decomposition of the real matrix $A$,
\begin{equation}
    A = QTQ^{T}, \label{eq:Schur}
\end{equation}
where $Q$ is an orthogonal matrix, and $T$ is a block upper triangular matrix \cite{horn1985} (as opposed to the unitary case in which then $T$ is always triangular). Then,
\begin{equation}
    H_{\mbox{ind}} = ||T - diag(T)||_F,
\end{equation}
with $diag(\cdot)$ being an $(n \times n)$ matrix containing only the diagonal elements of the matrix and $||\cdot||_F$ is the Frobenius norm. For the matrices in Eq. \ref{eq:Sample2by2}, $H_{\mbox{ind}}= \alpha$. So in our fist example it may seem like $H_{\mbox{ind}}$ paired with the Jordan form may be enough for us characterize the transient behavior of the system. However as can be seen in Fig. \ref{fig:TwobyTwo_Ex2} this is not the case. Below are three matrices,
\begin{equation}
    A_1 = 
    \begin{pmatrix}
    -2 & 5 \\
    0 & -2
    \end{pmatrix},
    A_2 = 
    \begin{pmatrix}
    -1 & 5 \\
    0 & -3
    \end{pmatrix},
     A_3 = \frac{\max \Lambda(\mbox{Sym}(A_1))}{\max \Lambda(\mbox{Sym}(A_2))}A_2,
    \label{eq:Sample2by2_Ex2}
\end{equation}
where,
\begin{equation}
    \Lambda(B) = \{\lambda | \lambda \ \mbox{is an eigenvalue of} \ B\},
\end{equation}
is the set of all eigenvalues of $B$ and
\begin{equation}
    \mbox{Sym}(A) = \frac{A+A^T}{2},
\end{equation}
is the symmetric part of A.
We choose $H_{\mbox{ind}}(A_1) = H_{\mbox{ind}}(A_2) = 3.9$ here.  Also note that $tr(A_1) = tr(A_2)$, where $tr(\cdot)$ denotes the trace. So by this measure of non-normality these two matrices are indistinguishable, and yet $A_2$ clearly has a transient while $A_1$ does not. 

The idea of normalizing $H_{\mbox{ind}}$ as was done in \cite{muolo2021} using $\tilde{H}_{\mbox{ind}} = \frac{H_{\mbox{ind}}}{||A||_F} $ does not resolve the situation either, because though the matrices from Eq. \ref{eq:Sample2by2} do have $\tilde{H}_{\mbox{ind}}$ increasing toward $1$ as $\alpha \rightarrow \infty$, $A_2$ in Eq. \ref{eq:Sample2by2_Ex2} actually has a smaller $\tilde{H}_{\mbox{ind}}$ than does $A_1$, with $\tilde{H}_{\mbox{ind}}(A_1) = 0.8704,\tilde{H}_{\mbox{ind}}(A_2) = \tilde{H}_{\mbox{ind}}(A_3) = 0.8452$, when rounded to $4$ significant digits. 
\begin{figure}[htbp]
\includegraphics[width=0.5\textwidth]{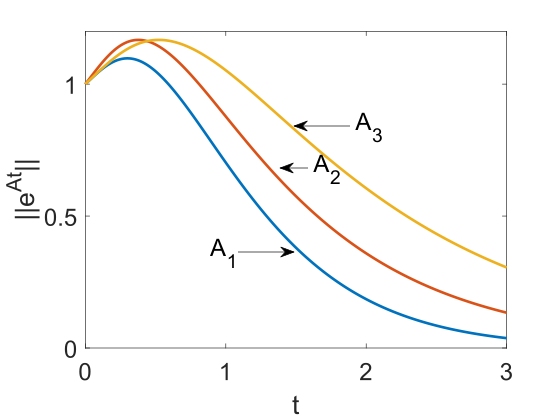}
\caption{In this case $5=H_{\mbox{ind}}(A_1) = H_{\mbox{ind}}(A_2) > H_{\mbox{ind}}(A_3)$, and $\tilde{H}_{\mbox{ind}}(A_1) = 0.8704$ (to $4$ significant digits), $\tilde{H}_{\mbox{ind}}(A_2) = \tilde{H}_{\mbox{ind}}(A_3) = 0.8452$. Finally $\omega(A_1)=\omega(A_3)=0.5$, and $\omega(A_2) = 0.6926$ (again to $4$ significant digits). All three matrices also have the same Jordan form and so the sensitivity index $\mathbb{S}$ is the same. \label{fig:TwobyTwo_Ex2}}
\end{figure}
An additional scalar which may directly get at the size of the transient for linear dynamics, is known as the numerical abscissa ($\omega(A)$), which we will now define:
\begin{equation}
    \omega(A) = \max(\Lambda(\mbox{Sym}(A))).
\end{equation}
$\omega(A)$ controls the behavior of $||e^{At}||, \ \mbox{as} \ t \rightarrow 0$ \cite{trefethen1997}. Thus $\omega(A)$ divulges what the initial slope will be. However  as can be seen in Fig. \ref{fig:TwobyTwo_Ex2}, though this reveals the initial slope, it gives no information about the size of the transient, as the dynamics generated from $A_2$ and $A_3$ have transients of the same size, but the one from $A_2$ has the larger initial slope.

So now it has been established with simple $(2 \times 2)$ examples that the transient dynamics of a non-normal system are related to the Jordan normal form, as well as the numerical abscissa and Henrici's index. These together are not enough to tell the complete story of the transient. So now we will focus on another tool, the pseudospectra.
\section{Pseudospectra \label{sec:Psuedospectra}}
Pseudospectra is a generalization designed for spectral analysis of non-normal matrices (though the pseudospectra is well defined for normal matrices). As noted in \cite{trefethen2005} pseudospectra have been known by many names, for instance they are called approximate eigenvalues in \cite{boccara1990}. Throughout this work we will use the name pseudospectra. The $\epsilon$-pseudospectra has several equivalent definitions \cite{trefethen2005}, however we will use the following:
\begin{equation}
    \Theta_{\epsilon}(A) = \{z \in \mathbb{C} \ | \ \sigma_n(zI-A) < \epsilon   \},
\end{equation}
where $\sigma_n(B)$ is the smallest singular value of $B$ and $I$ is the identity matrix. By examining the resulting level sets in the complex plane, we may then regain the ability to determine the stability of a non-normal
matrix. 

For normal matrices the stability is determined by the largest eigenvalue, if that eigenvalue is negative then for linear dynamics the system is stable, and the dynamics from any initial condition will eventually settle into a fixed point. This is also true for non-normal matrices, however the size and duration of the transient of a non-normal matrix is related to the smallest $\epsilon$ for which $\Theta_{\epsilon}$ crosses into the right side of the complex plane. The relevant fact to our interest here is that a lower bound on the transient is given by $\frac{\beta_\epsilon}{\epsilon}$, and importantly this quantity can be arbitrarily large \cite{trefethen2005}, where,
\begin{equation}
    \beta_\epsilon = \sup \{\mbox{Re}(z)| z \in \Theta_\epsilon(A)\},
\end{equation}
is called the $\epsilon$-psuedospectral abscissa.

\section{Optimal Networks and Synchronization}
We begin our discussion of optimal networks by reviewing the definition of  the master stability function (MSF) \cite{pecora1998}. In this seminal work, Pecora and Carroll showed that a model of the form:
\begin{equation}
    \dot{\bm{x}_i} = f(\bm{x}_i) + \sum \limits_j [\mathcal{A}]_{ij}
    h(\bm{x}_i,\bm{x}_j), \label{eq:DynamicalModel}
\end{equation}
where $\mathcal{A}$ is the adjacency matrix, can be linearized around the synchronous state, and amazingly that one only needs to examine the exponential stability of individual blocks of the form
\begin{equation}
    \dot{\eta} = [\mathcal{J}(f) + k\lambda_i \mbox{DH}]\eta, \label{eq:MSFBlock}
\end{equation}
where $\mathcal{J}(\cdot)$ is the (possibly) time dependent Jacobian, $\lambda_i$ is an eigenvalue of the graph Laplacian ($L$) and DH is the derivative of the coupling function \cite{pecora1998}. The graph Laplacian is given by:
\begin{equation}
    L = D-\mathcal{A},
\end{equation}
such that,
\begin{equation}
    [D]_{ij} = 
    \begin{cases}
    0 \ \mbox{if} \ i\neq j \\
    \sum_j [\mathcal{A}]_{ij} \ \mbox{if} \ i = j
    \end{cases}.
\end{equation}
Since we  are interested in the eigenvalues of $L$ it will be useful to note here that they may be sorted as follows:
\begin{equation}
    0= \lambda_1 \leq \mbox{Re}(\lambda_2) \leq ... \leq \mbox{Re}(\lambda_n). 
\end{equation}
In the context of Eq. \ref{eq:MSFBlock}, $\lambda_1$ represents the synchronization manifold, and $\lambda_i \ (\forall i \neq 1)$ represent directions transverse to that manifold. Pecora and Carroll found that now the exponential stability of the synchronization manifold could be examined by use of the master stability equation:
\begin{equation}
    \dot{\xi} = [\mathcal{J}(f) + (a+bi)DH] \xi. \label{eq:MSE}
\end{equation}

\begin{figure}[htbp]
\includegraphics[width=0.5\textwidth]{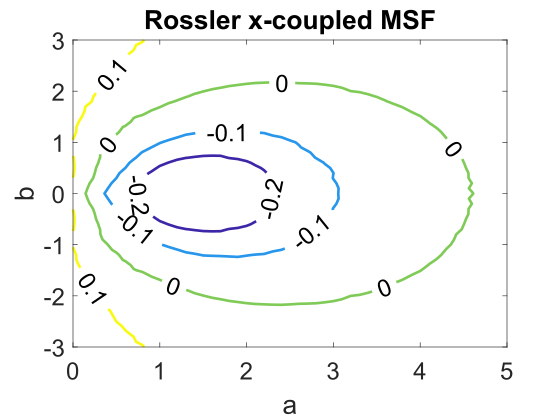}
\caption{The level sets of the MSF for an x-coupled R{\"o}ssler system. The system has an exponentially stable synchronization manifold if all of the nonzero eigenvalues of the graph Laplacian fall within the $0$ contour.  \label{fig:MSF_Contour}}
\end{figure}

The MSF is then defined by determining the maximum Lyapunov exponents of Eq. \ref{eq:MSE} for each $a,b$ pair. If the MLE is negative, then the state is exponentially stable, and so if the eigenvalues of the graph Laplacian all lie inside the $0$ level set (shown in Fig. \ref{fig:MSF_Contour}) of the MSF, then the synchronous state is declared exponentially stable. 

We now review the concept of optimal networks as they were called by Nishikawa and Motter \cite{nishikawa2006a,nishikawa2006b,nishikawa2010}. In their work they defined optimal networks to be networks which minimized the normalized spread of the eigenvalues. That is defining the normalized spread ($\rho$) to be,
\begin{equation}
    \rho^2 = \frac{1}{d^2(n-1)} \sum \limits_{i=2}^n |\lambda_j - \bar{\lambda}|, 
\end{equation}
where $d = \frac{1}{n}\sum \limits_{i} \sum \limits_{j \neq i} \mathcal{A}_{ij}$ is the average degree of the network and $\bar{\lambda} = \frac{1}{n-1}\sum \limits_{i=2}^n \lambda_i$ is the average of the nonzero eigenvalues of the graph Laplacian. Note that $\rho^2 = 0$ only if,
\begin{equation}
   \lambda_2 = \lambda_3 ... =\lambda_n \label{eq:optimalnetcondition} 
\end{equation}
and so a network is termed optimal if the condition set in Eq. \ref{eq:optimalnetcondition} is satisfied. Such networks are optimal in the sense that they exhibit the maximal range of coupling strengths under which a network can have an exponentially stable synchronous state, as determined by the MSF. All of the networks examined below are optimal.

\section{Laplacian Pseudospectra}
Of note, before beginning discussion on Laplacian pseudospectra, is that for the computation of the psuedospectra, a scaled identity matrix $zI$ is added to the matrix $-A$, and every eigenvalue of $-A$ will be changed by the amount $z$. To see this note that any square matrix may be written in Jordan form \cite{golub2013}, with $A = SJ(A)S^{-1}$ and thus,
\begin{equation}
    A+\kappa I = SJ(A)S^{-1} + \kappa I, \kappa \in \mathbb{C}. \label{eq:JordanSum1}
\end{equation}
Eq. \ref{eq:JordanSum1} implies that:
\begin{equation}
    S^{-1}(A+\kappa I) = J(A)S^{-1}+\kappa S^{-1} \implies S^{-1}(A+\kappa I)S = J(A)+ \kappa I. \label{eq:JordanSum2}
\end{equation}
Finally from Eq. \ref{eq:JordanSum2} it may be concluded that,
\begin{equation}
    A+\kappa I = S^{-1}(J(A)+\kappa I)S. \label{eq:JordanSum3}
\end{equation}
So from Eq. \ref{eq:JordanSum3}, and the fact that $J(A)$ is an upper triangular matrix, it is clear that all of the diagonal entries of $J(A)$, and thereby the eigenvalues as well, have been adjusted by the amount $\kappa$.
We are now ready to define the $\epsilon$-Laplacian pseudospectra {\it of a graph Laplacian} as:
\begin{equation}
    \Phi_{\epsilon}(L) = \{z \in \mathbb{C} \ | \ \sigma_{n-1}(E-L)< \epsilon, E = z\mathbb{I}_{n-1}\},
\end{equation}
where $E$ is an $(n \times n)$ matrix which perturbs all of the eigenvalues of $L$ by a constant amount $z \in \mathbb{C}$ except for one of the zero eigenvalues, $\mathbb{I}_{n-1}$ is the identity matrix with the the first entry set to $0$, $\mathbb{I}_{1,1} = 0$, $\sigma_{n-1}(\cdot)$ is the second smallest singular value of the matrix, and the input matrix $L$ is assumed to be a graph Laplacian. 

For numerical computation of $\Phi_{\epsilon}$,the Schur decomposition will be used, which was already introduced in Eq. \ref{eq:Schur}. The graph Laplacian $L$ may be written $L = UT_1U^*$, where $U$ is a unitary matrix.
A sub-matrix $T_2$ of $T_1$ is then created by keeping all of the rows and columns, except for one containing a zero eigenvalue. Finally we set $\Phi_{\epsilon}=\Theta_{\epsilon}(T_2)$. For reproducibility we make code available at \cite{fish2022soft}.

In Fig. \ref{fig:LapPseudo} an example of the $\epsilon$-Laplacian pseudospectra is shown for $3$ different networks (a $6$ node example of the networks is given in Fig. \ref{fig:ToyNetworkBasinStability}(d)), with the parameters $\gamma = 0$ in Fig. \ref{fig:LapPseudo}(a), $\gamma=0.5$ in Fig. \ref{fig:LapPseudo}(b) and $\gamma=1$ in Fig. \ref{fig:LapPseudo}(c). It is clear from this figure that not all optimal networks are created equal, some are more likely to desynchronize due to a transient than others.
\begin{figure}[htbp]
\includegraphics[width=0.95\textwidth]{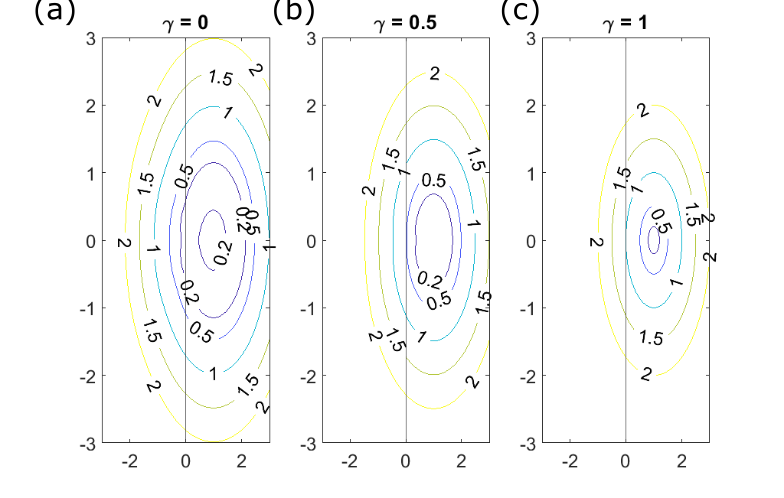}
\caption{Here the Laplacian pseudospectra of the toy network described in Sec. \ref{sec:LPR} is shown for three values of $\gamma$, (a) $\gamma = 0$, (b) $\gamma = 0.5$ and (c)$\gamma=1$ . Each of these networks is optimal, but the network with $\gamma = 1$ needs a much larger perturbation to cross into the left (unstable) part of the plane, than the others.    \label{fig:LapPseudo}}
\end{figure}

\section{Laplacian Pseudospectral Resilience \label{sec:LPR}}
We now define the {\it Laplacian pseudospectral resilience} (LPR) of an optimal network. For $L$ note that $\lambda_2$ is known as the algebraic connectivity of the graph \cite{fiedler1967}. $\mbox{Re}(\lambda_2)$ controls the asymptotic rate of convergence to the synchronous state of any network, and thus in general it is desirable to maximize $\mbox{Re}(\lambda_2)$. For non-normal networks $\lambda_2$ also plays a role in the size of the transient. As noted in Section \ref{sec:Psuedospectra}, the lower bound for the size of the transient is related to $\beta_\epsilon$, and so an analogous quantity is defined for the $\epsilon$-Laplacian psuedospectra,
\begin{equation}
    \zeta_\epsilon(L) = \inf \{\mbox{Re}(z)|z \in \Phi_\epsilon\},
\end{equation}
which we call the $\epsilon$-Laplacian pseudospectral abscissa. In this case $\inf$ is chosen rather than $\sup$ since the stability of the synchronous state is related to $-L$. For normal graph Laplacians, $\zeta_{\lambda_2} = 0$, yet in the non-normal case $\zeta_{\lambda_2}$ may cross into the left plane. This leads to the introduction of LPR, with,
\begin{equation}
    \mbox{LPR} = \max (-\zeta_{\lambda_2}(L),-\zeta_{\lambda_2}(L^T)) \in [0, \infty). 
    \label{eq:LPR}
\end{equation}
Note that in non-normal graph Laplacian the left and right eigenvectors may not coincide, hence the choice of the maximum between $-\zeta_{\lambda_2}(L)$ and $-\zeta_{\lambda_2}(L^T)$ for LPR in Eq. \ref{eq:LPR}.

LPR grows along with the lower bound on the transient, so LPR should be directly proportional to the stability of any non-normal graph Laplacian in systems where $\lambda_2$ represents the synchronizability of the system (the synchronizability of some systems may be much more complicated \cite{huang2009generic} and thus study of such scenarios is reserved for future work). As will be shown in the next paragraph, numerical evidence suggests that this is indeed true. We note that LPR may not be as useful for comparison between networks which are not optimal, or between networks which have different $\lambda_2$. Additionally we assume that the perturbations away from the synchronization manifold is finite but not too large. Such scenarios may require examination of the full Laplacian pseudospectra which is reserved for future work.

To examine the effects of non-normality on synchronization a toy network model was chosen as seen in Fig. \ref{fig:ToyNetworkBasinStability} (d) for a $6$ node example. This network can interpolate between two optimal networks, when $\gamma = 1$, the network is a directed star and when $\gamma = 0$ it is the directed chain. For any value of $\gamma$ the graph Laplacian has the eigenvalues $\lambda_j = 1 \ (\forall j >1)$, meaning the network remains optimal regardless of the value of $\gamma$. We will restrict analysis to $\gamma \in [0,1]$ for this work. For $\gamma \in [0,1]$, $\tilde{H}_{\mbox{ind}}$ changes with $\gamma$ and reaches a minimum value when $\gamma =0.5$. When $\gamma = 1$, the graph Laplacian is diagonalizable, but for $\gamma < 1$ there is a Jordan block of size $n-1$ for the eigenvalue $1$.
\begin{figure}[htbp!]
\includegraphics[width=0.95\textwidth]{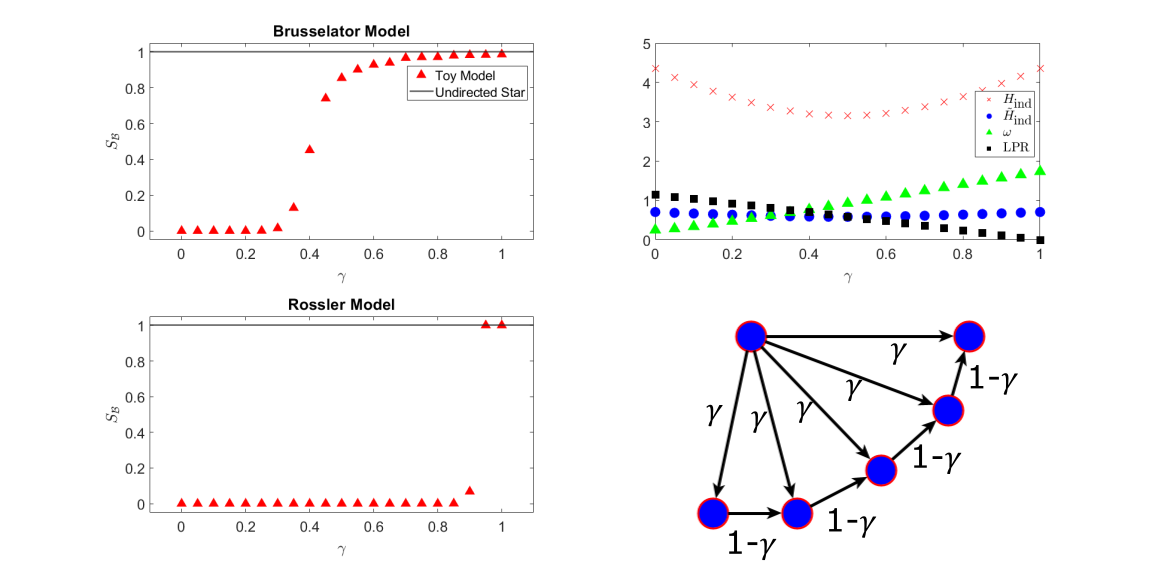}
\caption{(a) We show the basin stability versus the parameter $\gamma$ from our toy model. While $\tilde{H}_{\mbox{ind}}$ predicts that the network should be most stable at $\gamma = 0.5$, we can see that in fact (as predicted by our LPR) the stability monotonically decreases with as $\gamma \rightarrow 0$. Thus the choice of which optimal network matters, for instance choosing the optimal network with $\gamma = 1$ (that is the directed star network) gives nearly the same level of stability as the undirected star network, whereas choosing the network with $\gamma = 0$ produces a high level of instability which would not be predicted from master stability analysis.For simulation purposes, $e = 2.5, g = 1.12, k_x = 0.7, k_y = 5, n =20$ and to obtain an estimate of $\mathcal{R}$ we averaged over $5000$ initial conditions, integrating for $4000$ time units with $\delta = 0.1$. (b) The same as (a) except using the R{\"o}ssler x-coupled system with parameters $p = q = 0.2$, $r = 7$, $k_x = 4.45$ from Eq. \ref{eq:RosslerDynamics}. (c) $H_{\mbox{ind}}, \tilde{H}_{\mbox{ind}}, \omega,$ and LPR are shown for the various values of $\gamma$. For all of these scalars the lower the value the higher the predicted stability, note that both $H_{\mbox{ind}}$ and $\tilde{H}_{ind}$ predict that the most stable network is with $\gamma = 0.5$, and $\omega$ predicts that $\gamma = 0$ should be the most stable, only LPR correctly orders all of the networks by stability. (d) A 6 node example of the toy network we used for the simulations.  \label{fig:ToyNetworkBasinStability}}
\end{figure}
For comparison of the effectiveness of the various scalars for predicting the stability of the system, two models, the R{\"o}ssler x-coupled system following \cite{fish2017} and the Brusselator model following \cite{muolo2021}, are used. In both cases diffusive coupling will be assumed, that is $h(x,y) = y-x$ from Eq. \ref{eq:DynamicalModel}. For the Brusselator model \cite{prigogine1968} the equations describing the dynamics are:
\begin{equation}
    \begin{bmatrix}
    \dot{x}_i \\
    \dot{y}_i
    \end{bmatrix} =
    \begin{bmatrix}
    1+(e+1)x_i+gx_i^2y_i-k_x \sum \limits_{j=1}^n L_{ij}x_j \\
    ex_i - gx_i^2y_i -k_y \sum \limits_{j = 1}^n L_{ij}y_j
    \end{bmatrix} \ (\forall i \in \{1,2,...,n\}).\label{eq:BrussDynamics}
\end{equation}
The equations for the R{\"o}ssler system \cite{rossler1976} are as follows,
\begin{equation}
    \begin{bmatrix}
    \dot{x}_i \\
    \dot{y}_i \\
    \dot{z}_i
    \end{bmatrix} = 
    \begin{bmatrix}
    y_i - z_i - k_x \sum \limits_{j=1}^n L_{ij}x_j\\
    x_i + py_i \\
    q + z_i(x_i - r)
    \end{bmatrix} \ (\forall i \in \{1,2,...,n\}).\label{eq:RosslerDynamics}
\end{equation}
For the purposes of simulation, initially the system is assumed to be synchronized and to lie on the attractor, so each state is started with a random initial condition and a single node,
\begin{equation}
   \bm{x}^{0} \in \mathcal{U}([0,1]^{\mathcal{D}})
\end{equation}
 where $\mathcal{D}$ is $2$ for Eq. \ref{eq:BrussDynamics} and $3$ for Eq. \ref{eq:RosslerDynamics}. Then either Eq. \ref{eq:BrussDynamics} or Eq. \ref{eq:RosslerDynamics} are integrated for $4000$ time units in order to allow the system to settle onto the attractor. $\bm{x}^{4000}$ is assumed to be the synchronous state for each node, and so to obtain a perturbed initial condition we add a perturbation of size $\delta$, and the initial condition is given by,
 \begin{equation}
 \bm{X}^{0} = 
     \begin{bmatrix}
     \bm{x}^{4000}\\
     \bm{x}^{4000} \\
     \vdots \\
     \bm{x}^{4000}
     \end{bmatrix} +
     \delta \Delta \in \mathbb{R}^{n\mathcal{D}}, \label{eq:InitialConditions}
 \end{equation}
 where,
 \begin{equation}
     \Delta = \frac{v}{||v||},
 \end{equation}
 and
 \begin{equation}
     v \in \mathcal{N}(0,1)^{n\mathcal{D}}.
 \end{equation}

To estimate the resilience to perturbation a quantity known as basin stability \cite{menck2013,schultz2017} ($S_\mathcal{B}$) 
is defined as,
\begin{equation}
    S_\mathcal{B} = \int \limits_R \bm{1}_{\mathcal{B}} \ d\mu \in [0,1].
\end{equation}
Here $\mathcal{B}$ is the basin of attraction, $\mu$ is the invariant measure, $R$ is a set containing the basin (that is $\mathcal{B} \subset R$)  and $\bm{1}$ is the indicator function.   

In Fig. \ref{fig:ToyNetworkBasinStability}(a) simulations are performed for the Brusselator system, with the parameters $e = 2.5, g = 1.12, k_x = 0.7, k_y = 5, \delta = 0.1, n = 20$. To obtain an estimate of $S_{\mathcal{B}}$, $5000$ initial conditions are sampled as in Eq. \ref{eq:InitialConditions}, and the system is declared to be synchronous if after integration of $4000$ time units and
\begin{equation}
    \mbox{sd}_x \leq 0.1,
\end{equation}
where $\mbox{sd}_x$ is the standard deviation of the x-component.
It can be seen that as $\gamma \rightarrow 0$, fewer and fewer initial conditions synchronize, that is $S_\mathcal{B} \rightarrow 0$. When $\gamma = 1$ approximately the same number of initial conditions reach a synchronous state as the undirected star (which is normal). This is despite the fact that the minimum value of $\tilde{H}_{\mbox{ind}}$ occurs at $\gamma = 0.5$. Indeed all scalar measures which have been previously introduced, fail to correctly order the optimal networks by $S_{\mathcal{B}}$ except for LPR. We note that,
\begin{equation}
\mathbb{S} = 
    \begin{cases}
    1 \ \mbox{if} \ \gamma = 1 \\
    n-1 \ \mbox{otherwise}
    \end{cases},
\end{equation}
and $\omega(-L)$ (all listed scalar measures are performed on $-L$, with the exception of LPR which was defined assuming $-L$ related to stability, since $L$ is positive semidefinite) has a minimum at $\gamma = 0$, predicting the least stable network of the group as the most stable. LPR approaches its minimum value (i.e. most stable) as $\gamma \rightarrow 1$, thereby correctly ordering the optimal networks by their basin stability. This is true in the R{\"o}ssler x-coupled case as well, though with the parameters  $p = q = 0.2$, $r = 7$, $k_x = 4.45$ from Eq. \ref{eq:RosslerDynamics}, and the integration time set to $6000$ time units instead of $4000$. Interestingly LPR allows for accurate predictions in this scenario, even though the synchronizability of the system is related to both $\lambda_2$ and $\lambda_n$ rather than just $\lambda_2$. However we are only examining {\it optimal} networks here, which have $\lambda_2 = \lambda_n$, which allows LPR to be useful in this case. To reiterate noted above, in the more general setting the full Laplacian pseudospectra will need to be examined.
\begin{figure}[htbp]
\includegraphics[width=\textwidth]{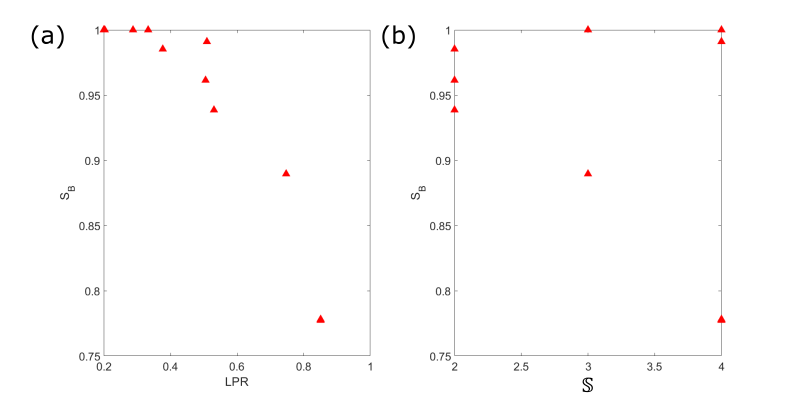}
\caption{Basin stability in optimal networks of different density. Displayed here are 11 randomly chosen 5 node optimal networks, selected to have $\mathbb{S} \in \{2,3,4\}$ but with differing densities of edges. Each of these graph Laplacians were scaled to have the same value of $\lambda_2$. In (a) it can be seen that LPR sorts these networks in terms of basin stability, those with the largest basins having smallest LPR as expected. In (b) it can be seen that the largest Jordan block is not useful in this context, in fact the networks with {\it both} the largest and the smallest $S_{\mathcal{B}}$ have the same value of $\mathbb{S}$ in this case. \label{fig:FinalExample}}
\end{figure}

Finally in Fig. \ref{fig:FinalExample} a comparison among optimal networks with different densities. In this situation the values of $\lambda_2$ may be different, thus a scalar $\mu = \frac{1}{\lambda_2}$ is introduced so that $\lambda_2^{(\mu)} = 1$, where $\lambda_2^{(\mu)}$ is an eigenvalue of $\mu L$. This allows for consideration of a scenario which was not explored in \cite{fish2017}, which compared only optimal networks of the same density. The 11 different networks chosen at random used in Fig. \ref{fig:FinalExample} are all 5 node optimal networks (the complete set of which were summarized in \cite{fish2017}), with varying density with values of $\mathbb{S}$ between 2 and 4. In this case the measure $\mathbb{S}$ completely fails to appropriately order the networks, as can be seen in Fig. \ref{fig:FinalExample} (b). Even though all of the examined networks are scaled to have the same value of $\lambda_2$, the size of the largest Jordan block is not a helpful comparison for the size of the transient. However LPR again gets the ordering correct, even in this scenario as seen in Fig. \ref{fig:FinalExample} (a). Fittingly these results seem to agree with a statement once made by Charles Van Loan in a technical report \cite{van1975study} “...one of the most basic tenets of numerical algebra, namely, anything that the Jordan decomposition can do, the Schur
decomposition can do better!”

\section{Conclusions}
In this work we have examined the stability of the synchronous state of "optimal" non-normal matrices, specifically of non-normal graph Laplacians. For such networks, MSF analysis is no longer useful in some circumstances, since the synchronization basin may shrink to an arbitrarily small size and we must appeal to other measures to characterize the stability to small but finite perturbations. We have shown that various scalar measures fail to properly characterize the behavior of the system in this scenario. We have introduced a new concept of Laplacian pseudospectra and a new measure which we have named Laplacian pseudospectral resilience (LPR). Finally we have provided numerical evidence that LPR outperforms other existing scalar quantities in ordering optimal networks by their basin stability when perturbations are small. 

\begin{acknowledgments}
We would like to thank Jie Sun for the many productive conversations we had on the subject of non-normal matrices and optimal networks.  E.M.B. has received funding from the Army Research Office (ARO) (No. N68164-EG) and the Office of Naval Research (ONR)  and the NIH-CRCNS and E.M.B. and J.F. were supported by the Defense Advanced Research Projects Agency (DARPA). 
\end{acknowledgments}

\section*{Data Availability Statement}
Data available on request from the authors

\appendix
\section{Jordan form as it relates to transient growth}
\label{ap:1}
A matrix $B \in \mathbb{C}^{n \times n}$ is called {\it diagonalizable} if there exists an invertible matrix $S$ such that \cite{golub2013}:
\begin{equation}
    B = S\Lambda S^{-1}, \label{eq:Diagonalizable}
\end{equation}
where 
\begin{equation}
\Lambda =
\begin{bmatrix}
\lambda_1 & 0 & ... & 0\\
0 & \lambda_2 & ... & 0 \\
\vdots & \vdots & \vdots & \vdots \\
0 & ... & ... & \lambda_n
\end{bmatrix}
\end{equation}
and $\lambda_i \in \mathbb{C} \ (\forall i \in \{1,2,...,n\})$ are the eigenvalues of $B$.
\begin{theorem}
Let $B$ be diagonlizable and $\boldsymbol{\eta} = S^{-1}\boldsymbol{x}$. Then the solution to 
\begin{equation}
    \dot{\boldsymbol{\eta}} = \Lambda \boldsymbol{\eta} \label{eq:Transformed}
\end{equation}
is
\begin{equation}
    \boldsymbol{\eta} = 
    \begin{bmatrix}
    c_1 e^{\lambda_1 t} \\
    \vdots \\
    c_n e^{\lambda_n t}
    \end{bmatrix}
    \label{eq:Solution1}
\end{equation}
\end{theorem}
\begin{proof}
We begin with 
\begin{equation}
    \dot{\boldsymbol{x}} = B\boldsymbol{x}. \label{eq:LinearSystem}
\end{equation}
Since $B$ is diagonalizable we may substitute Eq. \ref{eq:Diagonalizable} into Eq. \ref{eq:LinearSystem} and we have:
\begin{equation}
    \dot{\boldsymbol{x}} = S \Lambda S^{-1}\boldsymbol{x}. \label{eq:Untransformed}
\end{equation}
Using a transformation of variables we find that Eq. \ref{eq:Untransformed} may be rewritten as Eq. \ref{eq:Transformed}. Since $\Lambda$ is a diagonal matrix we have:
\begin{equation}
\begin{bmatrix}
\dot{\eta_1} \\
\dot{\eta_2} \\
\vdots \\
\dot{\eta_n}
\end{bmatrix} =
    \begin{bmatrix}
    \eta_1 \lambda_1 \\
    \eta_2 \lambda_2 \\
    \vdots \\
    \eta_n \lambda_n
    \end{bmatrix}.
    \label{eq:EtaEquation1}
\end{equation}
We note that there are no "cross" terms in Eq. \ref{eq:EtaEquation1} each equation may then be solved independently. It is well known that the solution of 
\begin{equation}
    \dot{\eta_i} = \lambda_i \eta_i
\end{equation}
is
\begin{equation}
    \eta_i = e^{\lambda_i t}e^{C_i} = c_i e^{\lambda_i t}
\end{equation}
\cite{zill2012}, which allows us to recover Eq. \ref{eq:Solution1} completing the proof.
\end{proof}
However not all square matrices are diagonalizable. The following has also been derived in \cite{nishikawa2006b}.
The {\it Jordan Form} or {\it Jordan Canonical Form} of a matrix is given by \cite{golub2013}:
\begin{equation}
J = \begin{bmatrix}
J_1 & 0 & ... & 0 \\
0 & J_2 & ... & 0 \\
\vdots & \vdots & \vdots & \vdots \\
0 & ... & 0 & J_m
\end{bmatrix} \in \mathbb{C}^{n\times n},
\end{equation}
where
\begin{equation}
    J_i = 
    \begin{bmatrix}
    \lambda_i & 1 & ... & 0 & 0 \\
0 & \lambda_i & 1 & ... & 0 \\
0 & 0 & \lambda_i & 1 & 0 \\
\vdots & \vdots & \vdots & \vdots & \vdots \\
0 & ... & 0  & 0 & \lambda_i
    \end{bmatrix} \in \mathbb{C}^{k_i \times k_i}
\end{equation}
is a bi-diagonal matrix with all $1$'s on the upper diagonal. 
Without loss of generality we may order the Jordan blocks so that $k_1 \geq k_2 \geq ... \geq k_m$. We note that if $k_1 = 1$ then $m = n$ and the corresponding Jordan form is diagonal. Any square matrix $D \in \mathbb{C}^{n \times n}$ may be written as: 
\begin{equation}
    D = SJS^{-1},
\end{equation}
for invertible $S$ \cite{golub2013}. Following the formulation above, we may rewrite the equation
\begin{equation}
    \dot{\boldsymbol{x}} = D\boldsymbol{x},
\end{equation}
as
\begin{equation}
    \dot{\boldsymbol{\eta}} = J\boldsymbol{\eta}.\label{eq:JordanDiffEq}
\end{equation}
\begin{theorem}
The solution to Eq. \ref{eq:JordanDiffEq} is given by: 
\begin{equation}
    \begin{bmatrix}
    \eta_1 \\
    \vdots \\
    \eta_{k_1} \\
    \eta_{(k_1+1)} \\
    \vdots \\
    \eta_{(k_1+k_2)} \\
    \vdots \\
    \eta_{(k_1+...+k_{m-1}+1)} \\
    \vdots \\
    \eta_n
    \end{bmatrix}
    =
    \begin{bmatrix}
    [\frac{c_{k_1}t^{k_1 - 1}}{(k_1-1)!}+\frac{c_{(k_1-1)}t^{k_1 - 2}}{(k_1-2)!}+... + c_1]e^{\lambda_1 t}\\
    \vdots\\
    c_{k_1}e^{\lambda_1 t}\\
    [\frac{c_{(k_1+k_2)}t^{k_2 - 1}}{(k_2-1)!}+\frac{c_{(k_1+k_2-1)}t^{k_2 - 2}}{(k_2-2)!}+... + c_{(k_1+1)}]e^{\lambda_2 t}\\
    \vdots \\
    c_{(k_1+k_2)}e^{\lambda_2 t}\\
    \vdots \\
    [\frac{c_{n}t^{k_m - 1}}{(k_m-1)!}+\frac{c_{n-1}t^{k_m - 2}}{(k_m-2)!}+... + c_{(k_1+...+k_{m-1}+1)}]e^{\lambda_m t}\\
    \vdots \\
    c_{n}e^{\lambda_m t}\\
    \end{bmatrix}
\end{equation}
\end{theorem}
\begin{proof}
We may write the individual terms of Eq. \ref{eq:JordanDiffEq} as:

\begin{equation}
    \begin{bmatrix}
    \dot{\eta}_1\\
    \dot{\eta}_2 \\
    \vdots \\
    \dot{\eta}_{k_1}\\
    \dot{\eta}_{(k_1+1)}\\
    \vdots \\
    \dot{\eta}_{(k_1+k_2)}\\
    \vdots \\
    \dot{\eta}_{(k_1+...k_{m-1}+1)}\\
    \vdots \\
    \dot{\eta}_n
    \end{bmatrix}
=
    \begin{bmatrix}
    \lambda_1 \eta_1 + \eta_2\\
    \lambda_1 \eta_2 + \eta_3 \\
    \vdots\\
    \lambda_1 \eta_{k_1}\\
    \lambda_2 \eta_{(k_1+1)} + \eta_{(k_1+2)} \\
    \vdots \\
    \lambda_2 \eta_{(k_1+k_2)}\\
    \vdots \\
    \lambda_m \eta_{(k_1+...k_{m-1}+1)}+ \eta_{(k_1+...k_{m-1}+2)}\\
    \vdots \\
    \lambda_m \eta_n
    \end{bmatrix}
\end{equation}
Clearly this situation is more complicated than the diagonalizable case, as the solution to one differential equation in any Jordan block is related to the solution of the next differential equation. However the final differential equation of any block may be solved with separation of variables, as it depends only on a single variable. Thus we may build the solution for a particular block by noting the solution to the final equation of the block is given by:
\begin{equation}
    \eta_i = c_i e^{\lambda_b t}, \label{eq:Eta1}
\end{equation}
where $b \in \{1,2,...,m\}$ is the block number. Now the previous equation in the block is dependent on the solution to the final equation in the block and so we get a differential equation:
\begin{equation}
    \dot{\eta}_{i-1} = \eta_{i-1}\lambda_b + \eta_i = \eta_{i-1}\lambda_b + c_i e^{\lambda_b t}, \label{eq:Eta2}
\end{equation}
with the equality on the right determined after substitution for $\eta_i$ from Eq. \ref{eq:Eta1}. Equation \ref{eq:Eta2} may be solved by integrating factor, giving a solution of
\begin{equation}
    \eta_{i-1} = (c_i t + c_{i-1})e^{\lambda_b t}. \label{eq:Eta3}
\end{equation}
We may now work our way further up the block by substituting Eq. \ref{eq:Eta3} into the previous equation obtaining:
\begin{equation}
    \dot{\eta}_{i-2} =  \eta_{i-2}\lambda_b + (c_i t + c_{i-1})e^{\lambda_b t}, \label{eq:Eta4}
\end{equation}
we again may obtain solution to the above equation by the integrating factor method and find the solution to be:
\begin{equation}
    \eta_{i-2} = (c_i \frac{t^2}{2} + c_{i-1}t + c_{i-2})e^{\lambda_b t}. \label{eq:Eta5}
\end{equation}
We may now apply induction to find the solution to all equations in the block, which will be:
\begin{equation}
    \eta_{i-k} = (c_i \frac{t^{k-1}}{(k-1)!} + c_{i-1}\frac{t^{k-2}}{(k-2)!} + ...+ c_{i-k})e^{\lambda_b t}. \label{eq:Eta6}
\end{equation}
Applying Eq. \ref{eq:Eta6} to all of the Jordan blocks completes the proof.
\end{proof}


\end{document}